\documentclass{article}

\pdfoutput=1

\usepackage{palatino}
\usepackage{euler}
\usepackage{amsfonts}
\usepackage{graphicx}
\usepackage{latexsym}
\usepackage{amsmath}
\usepackage{amssymb}
\usepackage{amstext}
\usepackage{amsthm}
\usepackage{enumerate}
\usepackage{pstricks,pst-grad,pst-node}
\usepackage[dutch,english]{babel}
\usepackage{rotating}
\usepackage[colorlinks]{hyperref}

\usepackage{tikz}
\usetikzlibrary{arrows,automata,backgrounds,calc,chains,fadings,folding,decorations.fractals,fit,patterns,positioning,mindmap,shadows,shapes.geometric,shapes.symbols,through,trees,decorations.markings,decorations.text}
\input{main.def}

\title{
  Unique Normal Forms in \\
  \mbox{Infinitary Weakly Orthogonal Term Rewriting}\thanks{%
  An earlier version of this paper appeared in~\cite{la:roel:2009}.}
}

\author{J\"org Endrullis \and Clemens Grabmayer \and Dimitri Hendriks \\[.5ex] \and Jan Willem Klop}

\date{}

\begin{document}

\maketitle

\begin{dedication}
  Dedicated to Roel de Vrijer on the occasion of his 60th birthday.
\end{dedication}

\begin{abstract}
  
The theory of finite and infinitary term rewriting is extensively developed 
for orthogonal rewrite systems, but to a lesser degree for weakly orthogonal rewrite systems.
In this note we present some contributions to the latter case of weak orthogonality, 
where critial pairs are admitted provided they are trivial.

We start with a refinement of the by now classical Compression Lemma, 
as a tool for establishing infinitary confluence~$\CRinf$, 
and hence the infinitary normal form property~$\UNinf$,
for the case of weakly orthogonal TRSs that %are collapse-free
do not contain collapsing rewrite rules.

That this restriction of collapse-freeness is crucial,
is shown in an elaboration of a simple TRS which is weakly orthogonal, but has two collapsing rules.
It turns out that all the usual theory breaks down dramatically.

We conclude with establishing a positive fact:
the diamond property for infinitary developments for \woTRS{s}, 
by means of a detailed analysis initiated by \mbox{van Oostrom} for the finite case.
\end{abstract}

\section{Preliminaries}
An \emph{infinitary rewrite rule} is a pair $\pair{s}{t}$ with
$s \in \ter{\asig}$ and $t\in\iter{\asig}$ such that $s$ is not a variable and every variable
in $t$ occurs in $s$.
A rewrite rule  $\pair{s}{t}$ is \emph{left-linear} if no variable 
has more than one occurrence in $s$.

An \emph{infinitary term rewriting system (iTRS)} is a pair
$\pair{\asig}{R}$ consisting 
of a signature $\asig$ and a set $R$ of infinitary rewrite rules.
An iTRS is called \emph{weakly orthogonal}
if all rules are left-linear and all critical pairs $\pair{s_1}{s_2}$ are trivial ($s_1 \equiv s_2$).

As a preparation for Section~\ref{sec:confluence} we will prove the following lemma,
which is a refined version of the Compression Lemma in left-linear TRSs.
In its original formulation 
(e.g.\ see Theorem~12.7.1 on page~689 in~\cite{terese:2003}),
it states that strongly convergent rewrite sequences in left-linear TRSs
can be compressed to length less or equal to $\omega$.

\begin{lemma}[Refined Compression Lemma]
  \label{lem:compression}
  Let $R$ be a left-linear iTRS. % in which all rules have finite left-hand sides.
  Let $\aseq \funin s \to^\alpha_R t$ be a rewrite sequence, 
  $d$ the minimal depth of a step in $\aseq$, 
  and $n$ the number of steps at depth $d$ in $\aseq$.
  Then there exists a rewrite sequence $\aseq' \funin s \to^{\le \omega}_R t$
  in which all steps take place at depth $\ge d$, and where precisely
  $n$ steps contract redexes at depth~$d$.
\end{lemma}

\begin{proof}
  We proceed by transfinite induction on the ordinal length $\alpha$
  of rewrite sequences $\aseq \funin s \to^\alpha_R t$
  with $d$ the minimal depth of a step in $\aseq$, 
  and $n$ the number of steps at depth $d$ in $\aseq$.
  
  In case that $\alpha = 0$ %, as well as in case that $\alpha \le \omega$,
  nothing needs to be shown. 
  
  Suppose $\alpha$ is a successor ordinal.
  Then $\alpha = \beta + 1$ for some ordinal $\beta$, and
  $\aseq$ is of the form $s \to^{\beta} s' \to t$.
  Applying the induction hypothesis to $s \to^{\beta} s'$ 
  yields a rewrite sequence $s \to^{\gamma} s'$ of length $\gamma\le\omega$
  that contains the same number of steps at depth $d$, and no steps
  at depth less than $d$.
  
  If $\gamma < \omega$, then $s \to^\gamma s' \to t$ is a rewrite sequence of
  length $\gamma + 1 < \omega$, in which all steps take place at depth
  $\ge d$ and precisely $n$ steps %contract redexes 
								  at depth $d$.
  
  If $\gamma = \omega$, we obtain a rewrite sequence of the form 
  $s \equiv s_0 \to s_1 \to \ldots \to^\omega s_\omega \to t$.
  Let $\ell \to r \in R$ be the rule applied in the final step $s_\omega \to t$,
  that is, $s_\omega \equiv \contextfill{\acontext}{\subst{\asubst}{\ell}} \to \contextfill{\acontext}{\subst{\asubst}{r}} \equiv t$ 
  for some context $\acontext$ and substitution $\asubst$.
  Moreover, let $d_h$ be the depth of the hole in $\acontext$, and $d_p$ the depth of the pattern of $\ell$.
  Since the reduction $s_0 \to^\omega s_\omega$ is strongly convergent, 
  there exists $n \in \nat$ such that all rewrite steps in $s_n \to^\omega s_\omega$ have depth $> d_h + d_p $,
  and hence are below the pattern of the redex contracted in the last step
  $s_\omega \to t$.
  As a consequence, there exists a substitution $\bsubst$
  such that $s_n \equiv \contextfill{\acontext}{\subst{\bsubst}{\ell}}$, and
  since $s_n \equiv \contextfill{\acontext}{\subst{\bsubst}{\ell}} 
		 \to^\omega \contextfill{\acontext}{\subst{\asubst}{\ell}} \equiv s_\omega$,
  it follows that $\myall{x \in \vars{\ell}}{\funap{\bsubst}{x} \to^{\le \omega} \funap{\asubst}{x}}$.
  We now prepend the final step $s_\omega \to t$ to $s_n$, that is:
  $s_n \equiv \contextfill{\acontext}{\subst{\bsubst}{\ell}} \to \contextfill{\acontext}{\subst{\bsubst}{r}}$.
  Even if $r$ is an infinite term, this creates at most $\omega$-many copies of subterms $\funap{\bsubst}{x}$ 
  with reduction sequences $\funap{\bsubst}{x} \to^{\le \omega} \funap{\asubst}{x}$ of length $\le \omega$.
  Since these reductions are independent of each other there exists an interleaving 
  $\contextfill{\acontext}{\subst{\bsubst}{r}} \to^{\le \omega} \contextfill{\acontext}{\subst{\asubst}{r}}$
  of length at most $\omega$ (the idea is similar to establishing countability of $\omega^2$ by dovetailing).
  Hence we obtain a rewrite sequence
  $\aseq' \funin s \to^{\le \omega} t$,
  since
  $s \to^n s_n \equiv \contextfill{\acontext}{\subst{\bsubst}{\ell}}
   \to \contextfill{\acontext}{\subst{\bsubst}{r}} \to^{\le \omega} \contextfill{\acontext}{\subst{\asubst}{r}} \equiv t$.
  It remains to be shown that $\aseq'$ contains only steps at depth $\ge d$, 
  and that it has the same number of steps as the original sequence $\aseq$ 
  at depth $d$.
	% If $d_h > d$, then this follows from the induction hypo
  This follows from the induction hypothesis and
  the fact that all steps in $s_n \to^\omega s_\omega$ 
  have depth $> d_h + d_p$ and thus also all steps of the interleaving
  $\contextfill{\acontext}{\subst{\bsubst}{r}} \to^{\le \omega} \contextfill{\acontext}{\subst{\asubst}{r}}$
  have depth $> d_h + d_p - d_p = d_h \ge d$ 
  (the application of $\ell \to r$ can lift steps 
   at most by the pattern depth $d_p$ of $\ell$).
  
  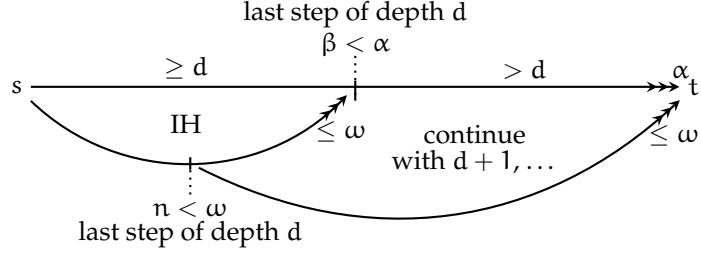
\begin{figure}[hpt!]
  \begin{center}
  \begin{tikzpicture}
	[inner sep=1mm,
	 node distance=45mm,
	 terminal/.style={
	   rectangle,rounded corners=2.25mm,minimum size=4mm,
	   very thick,draw=black!50,top color=white,bottom color=black!20,
	 },
	 >=stealth,
	 red/.style={thick,->>},
	 infred/.style={
	   thick,
	   shorten >= 1mm,
	   decoration={
		 markings,
		 mark=at position -3mm with {\arrow{stealth}},
		 mark=at position -1.5mm with {\arrow{stealth}},
		 mark=at position -0.001 with {\arrow{stealth}}
	   },
	   postaction={decorate}
	 }]
  
	\node (s) {$s$};
	\node (s') [right of=s] {};
	\node (sb) [node distance=1mm,below of=s] {\hphantom{$s$}};
	\node (s'b) [node distance=1mm,below of=s'] {};
	\node (t) [right of=s'] {$t$};
  
	\draw [infred] (s) -- (t) node [at end,above] {$\alpha$};
	\draw [thick] ($(s') + (0,.5ex)$) -- ($(s'b) + (0,-.5ex)$);
	\draw [thick,dotted] ($(s') + (0,2.5ex)$) -- ($(s'b) + (0,-.5ex)$) 
		  node [midway,yshift=5.2ex] {last step of depth $d$}
		  node [midway,yshift=2.7ex] {$\beta < \alpha$};
	\draw [infred] (s) to [bend right=45] (s') node [at end,below,xshift=-.5ex,yshift=-1.4ex] {$\le \omega$}
		  node [midway] (d) {};
  
	\node at ($(s)!.5!(s')$) [anchor=south] {$\ge d$};
	\node at ($(s')!.5!(t)$) [anchor=south] {$> d$};
  
	\draw [thick] ($(d) + (0,.5ex)$) -- ($(d) + (0,-.5ex)$);
	\draw [thick,dotted] ($(d) + (0,-2.5ex)$) -- ($(d) + (0,.5ex)$) 
		  node [midway,yshift=-4.7ex] {last step of depth $d$}
		  node [midway,yshift=-2.7ex] {$n < \omega$};
  
	\draw [infred] (d) to [bend right=35] (t) node [at end,below,xshift=-.5ex,yshift=-1.4ex] {$\le \omega$};
	\node at ($(s)!.5!(s')$) [yshift=-3ex] {IH};
	\node at ($(s')!.5!(t)$) [xshift=-4ex,yshift=-3.8ex] {continue};
	\node at ($(s')!.5!(t)$) [xshift=-4ex,yshift=-6ex] {with $d+1$, \ldots};
  \end{tikzpicture}
  \end{center}\vspace{-3ex}
  \caption{\textit{Compression Lemma, in case $\alpha$ is a limit ordinal.}}
  \label{fig:compression}
  \end{figure}
  
  Finally, suppose that $\alpha$ is a limit ordinal $> \omega$.
  We refer to Figure~\ref{fig:compression} for a sketch of the proof.
  Since $\aseq$ is strongly convergent, only a finite number of steps
  take place at depth $d$. Hence there exists $\beta < \alpha$ such that
  $s_{\beta}$ is the target of the last step at depth $d$ in $\aseq$.
  We have $s \to^\beta s_\beta \to^{\le \alpha} t$ 
  and all rewrite steps in $s_\beta \to^{\le \alpha} t$ are at depth $> d$.
  By induction hypothesis there exists 
  a rewrite sequence $\bseq \funin s \to^{\le \omega} s_\beta$
  containing an equal amount of steps at depth $d$ as $s \to^\beta s_\beta$.
  Consider the last step of depth $d$ in $\bseq$\,. 
  This step has a finite index $n < \omega$.
  Thus we have $s \to^* s_{n} \to^{\le \alpha} t$, 
  and all steps in $s_n \to^{\le \alpha} t$ are at depth $> d$. 
  By successively applying this argument to $s_n \to^{\le \alpha} t$ 
  we construct finite initial segments $s \to^* s_n$ 
  with strictly increasing minimal rewrite depth $d$. 
  Concatenating these finite initial segments 
  yields a reduction $s \to^{\le \omega} t$
  containing as many steps at depth $d$ as the original sequence.
\end{proof}

With this refined compression lemma we now prove that
also divergent rewrite sequences can be compressed to length less or equal
to $\omega$.

\begin{theorem}
  \label{thm:comp:div:seqs}
  Let $R$ be a left-linear iTRS. %  in which all rules have finite left-hand sides.
  For every divergent rewrite sequence $\aseq \funin s \to^\alpha_R$ of length $\alpha$
  there exists a divergent rewrite sequence
  $\aseq' \funin s \to^{\le\omega}_R$ of length less or equal to $\omega$.
\end{theorem}

\begin{proof}
  Let $\aseq \funin s \to^\alpha_R$ be a divergent rewrite sequence.
  Then there exist $k\in\nat$ such that infinitely many steps in $\aseq$
  take place at depth $k$.
  Let $d$ be the minimum of all numbers $k$ with that property.
  Let $\beta$ be the index of the last step above depth $d$ in 
  $\aseq$, $\aseq \funin s \to^\beta s_\beta \to^{\le \alpha}$.
  Then by Lemma~\ref{lem:compression} the rewrite sequence
  $s \to^\beta s_\beta$ can be compressed to a rewrite sequence
  $s \to^{\le \omega} s_\beta$ such that
  $s_\beta \to^{\le \alpha}$ consists only of steps at depth $\ge d$,
  among which infinitely many steps are at depth $d$.
  Let $n$ be the index of the last step of depth $\le d$ in the rewrite sequence $s \to^{\le \omega} s_\beta$.
  Then $s \to^* s_n \to^{\le \omega} s_\beta \to^{\le \alpha}$,
  and $s_n \to^{\le \omega} s_\beta \to^{\le \alpha}$ contains only steps at depth $\ge d$.
  Thus all steps with depth less than $d$ occur in the finite prefix
  $s \to^* s_n$.

  Now consider the rewrite sequence 
  $\aseq_1 \funin s_n \relcomp{\to^{\le \omega}}{\to^{\le \alpha}}$,
  say $\aseq_1 \funin s_n \to^\gamma$ for short,
  containing infinitely many steps at depth $d$.
  Let $\gamma'$ be the first step at depth $d$ in $\aseq_1$.
  Then $\aseq_1 \funin s_n \to^{\gamma'} u \to^{\le \gamma}$ for some term $u$
  and $s_n \to^{\gamma'} u$ can be compressed to $s_n \to^{\le \omega} u$ 
  containing exactly one step at depth $d$.
  Now let $m$ be the index of this step, then $s_n \to^m u' \to^{\le \omega} u \to^{\le \gamma}$
  where $s_n \to^m u'$ contains one step at depth $d$.
  Repeatedly applying this construction to $u' \to^{\le \omega} u \to^{\le \gamma}$
  we obtain a divergent rewrite sequence 
  $\aseq' \funin s \to^* s_n \to^* u' \to^* u'' \to \ldots$ 
  that contains infinitely many steps at depth $d$,
  and hence is divergent.
\end{proof}

\begin{remark}
  A slightly weaker version of Lemma~\ref{lem:compression}, as well as 
  Theorem~\ref{thm:comp:div:seqs} was formulated by the second author
  in a private communication with Hans Zantema. These statements have been published,
  in a reworked form, in~\cite{zant:2008} (see Lemma~3 and Theorem~4 there).
  The weaker version of Lemma~\ref{lem:compression}, which is sufficient
  to obtain a proof of Theorem~\ref{thm:comp:div:seqs}, states the following:
  Every strongly convergent rewrite sequence $\aseq \funin s \to^\alpha_R t$
  with $d$ the minimal depth of its steps can be compressed 
  to a rewrite sequence $\aseq' \funin s \to^{\le\omega}_R t$ of length $\le \omega$ 
  with the same or more steps at minimal depth $d$.

  We note that very closely related statements 
  have been formulated for infinitary CRSs (Combinatory Reduction Systems)
  by Jeroen Ketema in~\cite{kete:2008} (see Theorem~2.7 and Lemma 5.2 there).
\end{remark}

\section{Infinitary Unique Normal Forms}
\newcommand{\spterm}{$\ssuc\spre$-term}

In~\cite{klop:vrij:2005}, Klop and de Vrijer have shown 
that every orthogonal TRS exhibits the infinitary unique normal forms ($\UNinf$) property.
By way of contrast, we will now give a counterexample showing that the 
$\UNinf$ property does \emph{not} generalize to weakly orthogonal TRSs.
The counterexample is very simple: its signature consists of 
the unary symbols $\spre$ and $\ssuc$ with the reduction rules:
\begin{align*}
  \pre{\suc{x}} &\to x &
  \suc{\pre{x}} &\to x
  \punc.
\end{align*}
It is easily seen that this TRS is indeed weakly orthogonal.

Using $\ssuc$ and $\spre$ we have infinite terms such as 
$\ssuc\spre\ssuc\spre\ssuc\spre\ldots$,
where we drop the brackets that are associating to the right.
We write $\ssuc^\omega$ for $\ssuc\ssuc\ssuc\ldots$
and $\spre^\omega$ for $\spre\spre\spre\ldots$.
In fact,  $\ssuc^\omega$ and $\spre^\omega$ are the only infinite normal forms.

Given an infinite \spterm~$t$ we can plot in a graph 
the surplus number of $\ssuc$'s of $t$ when traversing 
from the root of $t$ downwards to infinity 
(or to the right if $t$ is written horizontally). % :)
This graph is obtained by counting $\ssuc$ for $+1$ and $\spre$ for $-1$.
We define $\funap{\mrm{sum}}{t,n}$ as the result of this counting up to depth $n$ in the term $t$.
For $t=\ssuc\spre\ssuc\spre\ssuc\spre\ldots$
the graph takes values, consecutively, $1,0,1,0,\ldots$, 
while for $t=\ssuc^\omega$ the excess number is $1,2,3,\ldots$,
for $t=\spre^\omega$ we have $-1,-2,-3,\ldots$.

\begin{figure}[ht!]
  \begin{center}
    \scalebox{.5}{\includegraphics{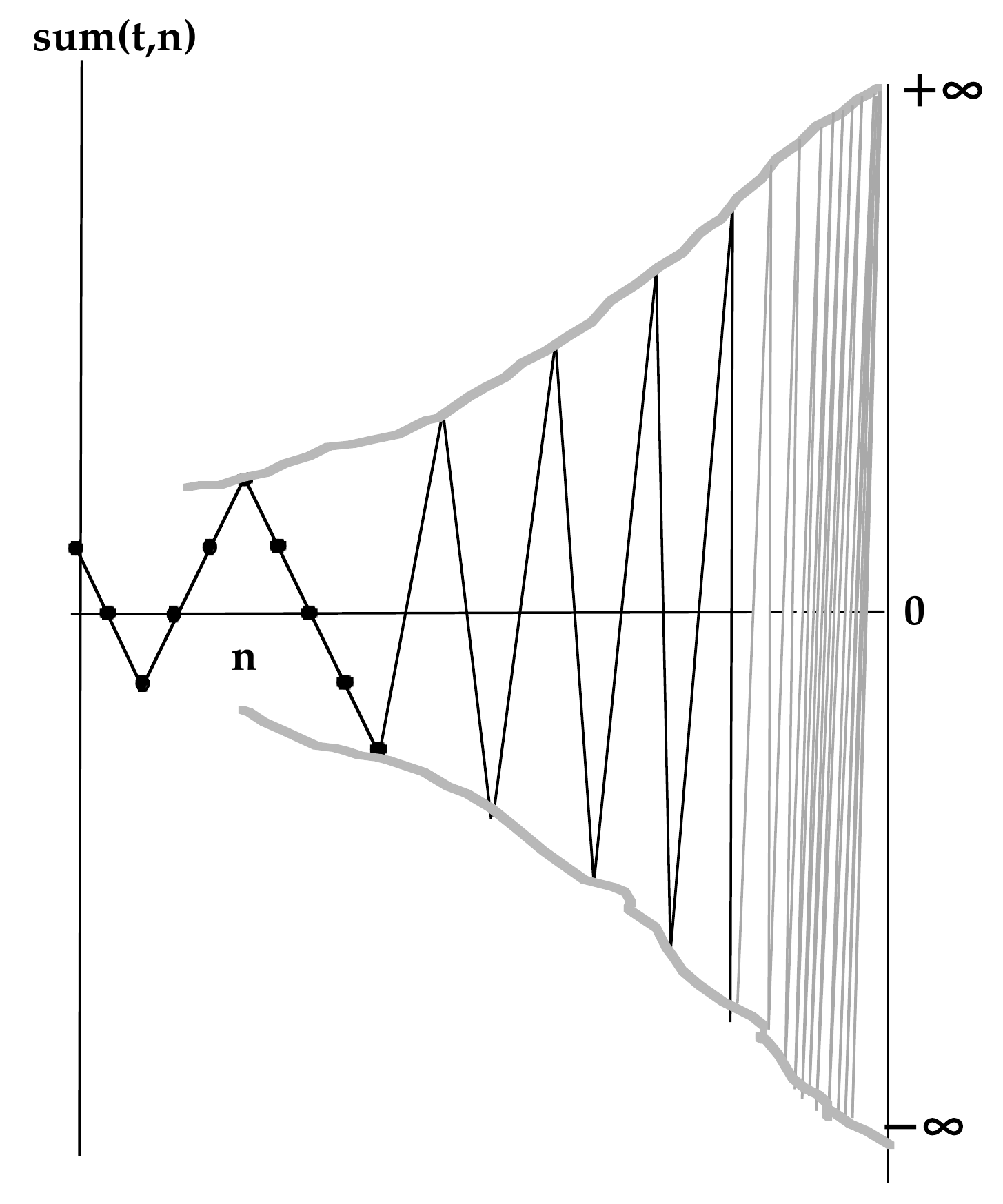}}
  \end{center}
  \caption{\textit{%
    Graph for the oscillating \spterm\ 
    $q = \ssuc^1\,\spre^2\,\ssuc^3\,\ldots$\,. 
  }}
  \label{fig:SP:graph}
\end{figure}

The upper and lower $\ssuc$-height of $t$ are defined as the supremum 
and infimum obtained by this graph,
i.e., $\sup_{n \in \nat} \funap{\mrm{sum}}{t,n}\}$ and $\inf_{n \in \nat} \funap{\mrm{sum}}{t,n}\}$,
respectively.
So the upper (lower) $\ssuc$-height of $(\ssuc\spre)^\omega$
is $1$ ($0$), of $\ssuc^\omega$ it is $\infinity$ ($0$), and of $\spre^\omega$ $0$ it is ($-\infinity$).

%\samepage{
Now we have:
\begin{proposition}
  \hfill
  \begin{enumerate}
    \item $t \to^{\omega} \ssuc^\omega$ if and only if the upper $\ssuc$-height of $t$ is $\infinity$,
    \item $t \to^{\omega} \spre^\omega$ if and only if the lower $\ssuc$-height of $t$ is $-\infinity$.
  \end{enumerate}
\end{proposition}
\begin{proof}
   Left for Roel.
\end{proof}
%}

Now let us take a term $q$ with upper $\ssuc$-height $\infinity$ \emph{and} lower $\ssuc$-height $-\infinity$\,!
Then $q$ reduces to both $\ssuc^\omega$ and $\spre^\omega$, both normal forms.
Hence $\UNinf$ fails.
Is there indeed such a $q$? 
Yes there is (see also Figure~\ref{fig:SP:graph}):
\begin{align*}
  q & =
  \ssuc \, \spre\spre \, \ssuc\ssuc\ssuc \, \spre\spre\spre\spre \, 
  \ssuc\ssuc\ssuc\ssuc\ssuc \, \spre\spre\spre\spre\spre\spre \, \ldots
  %\label{counterexample}
\end{align*}

To see that $q$ indeed reduces to both $\ssuc^\omega$ and $\spre^\omega$:
shifting the ``$\spre$-blocks to the right, so that they are `absorbed'
by the $\ssuc$-blocks, yields ever more $\ssuc$'s.
On the other hand, shifting the $\ssuc$-blocks to the right 
so that they are absorbed by the $\spre$-blocks, leaves infinitely many $\spre$'s.

%\begin{samepage}
The failure of $\UNinf$ for two collapsing rules raises the following question:
\begin{question}\label{q:one}
  What if we admit only \emph{one} collapsing rule? 
\end{question}
\noindent
We leave this question to future work.
%\end{samepage}

{
% roman font used only to let these symbols correspond more to the ones used in the picture
\newcommand{\A}{\mrm{A}}
\newcommand{\B}{\mrm{B}}
\newcommand{\RA}{\mrm{RA}}
We will now provide a more detailed analysis of various classes of \mbox{\spterm{s}}.
In Figure~\ref{fig:SP:Venn}, 
\begin{figure}[ht!]
  \begin{center}
    \scalebox{.5}{\includegraphics{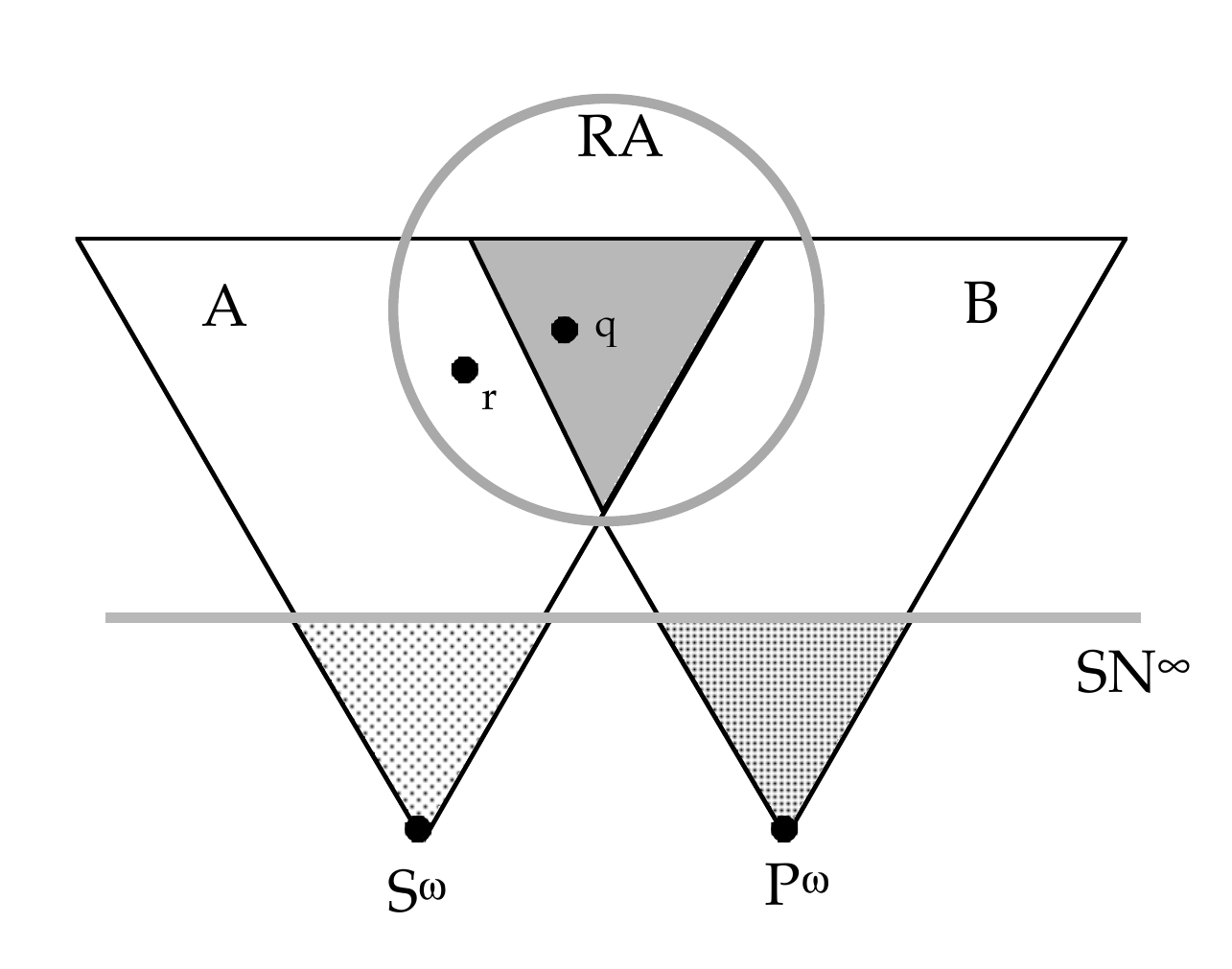}}
  \end{center}
  \caption{\textit{Classes of infinite \spterm{s}.}}
  \label{fig:SP:Venn}
\end{figure}
the extension of these classes is pictured.
Here $\A$ is the set of infinite \spterm{s} reducing to $\ssuc^\omega$,
and $\B$ that of those reducing to $\spre^\omega$ with a shaded non-empty intersection
containing the counterexample term~$q$ mentioned above.
The term $r = \ssuc \, \spre \, \ssuc^2 \, \spre^2 \, \ssuc^3 \, \spre^3 \, \ldots$
is an element of $(\A\setminus\B)\cap\RA$, where $\RA$ is the set of 
root-active terms.
The dotted part ${}\subseteq{\A\cup\B}$ is $\SNinf$.
The set $\RA$ is characterized as follows:
\newpage
\begin{proposition}
  \hfill
  \begin{enumerate}
    \item 
      An \spterm~$t$ is root-active if and only if the partial sums 
      $\funap{\mrm{sum}}{t,n}$ equal $0$ for infinitely many $n$.
      Equivalently:
    \item
      An \spterm~$t$ is root-active if and only if $t$ is the concatenation
      of infinitely many `finite zero words' $t_0,t_1,t_2,\ldots$. 
	  Here a zero word is a finite \spterm{} with the same number of $\ssuc$'s and $\spre$'s.  
	  If $w$ is a zero word, then $w$ clearly  reduces to the empty word.
  \end{enumerate}
\end{proposition}

\begin{proof}\label{prop:root-active:SPterms}
  The direction $\Leftarrow$ is obvious.

  For the direction $\Rightarrow$ 
  we label all $\ssuc$'s and $\spre$'s in the start term, say by numbering them
  from left to right, so e.g.\ the labelled $t$ could be:
  \[  
    \ssuc_0 \, \ssuc_1 \, \spre_2 \, \ssuc_3 \, \spre_4 \, \spre_5 \, \ldots 
  \]
  Then in a reduction of $t$ any $\ssuc$ or $\spre$ can be traced back to a unique ancestor in~$t$.
  Furthermore let $w_i$ be the prefix of $t$ of length $i$.
  Observation: if $\ssuc_i$ or $\spre_i$ gets at a root position in a reduction of $t$, 
  then $w_i$ is a $0$-word.
  Proof: easy.
\end{proof}

\begin{corollary}
  $\SNinf$ is the set of \spterm{s} 
  that are infinitarily strongly normalizing. 
  Then $t\in\SNinf$ if and only if each value $\funap{\mrm{sum}}{t,n}$
  for $n=0,1\dots$ occurs only finitely often, or equivalently
  $\lim_{n \to \infty} \funap{\mrm{sum}}{t,n}$ exists (then it is $\infty$ or $-\infty$).
\end{corollary}
\begin{proof}
  Note that then the normal form is unique,
  since there is no oscillation.
\end{proof}
}

\section{Infinitary Confluence}\label{sec:confluence}
In the previous section we have seen that the property $\UNinf$ 
fails dramatically for weakly orthogonal TRSs when collapsing rules are present,
and hence also $\CRinf$. 
Now we show that \woTRS{s} without collapsing 
rules are infinitary confluent ($\CRinf$),
and as a consequence also have the property $\UNinf$.

We adapt the projection of parallel steps in \woTRS{s}
from~\cite[Section~8.8.4.]{terese:2003} to infinite terms.
The basic idea is to orthogonalize the parallel steps,
and then project the orthogonalized steps.
The orthogonalization uses that overlapping redexes 
have the same effect and hence can be replaced by each other.
In case of overlaps we replace the outermost redex by the innermost one.
%\todo{we do not understand this! is there a contradiction between these sentences?}
This is possible for infinitary parallel steps since there
can never be infinite chains of overlapping, nested redexes (see Figure~\ref{fig:chain}).
For a treatment of infinitary developments where such chains can occur,
we refer to Section~\ref{sec:diamond}.
See further~\cite[Proposition~8.8.23]{terese:2003} for orthogonalization 
in the finitary case.

\begin{proposition}
  Let $\pstep{\astep}{s}{t_1}$, $\pstep{\bstep}{s}{t_2}$ be parallel steps in a \woTRS{}.
  Then there exists an \emph{orthogonalization $\pair{\astep'}{\bstep'}$ of $\astep$ and $\bstep$},
  that is, a pair of orthogonal parallel steps such that $\pstep{\astep'}{s}{t_1}$, $\pstep{\bstep'}{s}{t_2}$.
\end{proposition}

\begin{proof}
In case of overlaps between $\astep$ and $\bstep$, 
then for every overlap we replace the outermost redex by the innermost one
(if there are multiple inner redexes overlapping, then we choose the left-most among the top-most redexes).
If there are two redexes at the same position but with respect to different rules,
then we replace the redex in $\bstep$ with the one in $\astep$.
See also Figure~\ref{fig:orthogonalization:parallel}.
\begin{figure}[hpt!]
\begin{center}
  \includegraphics{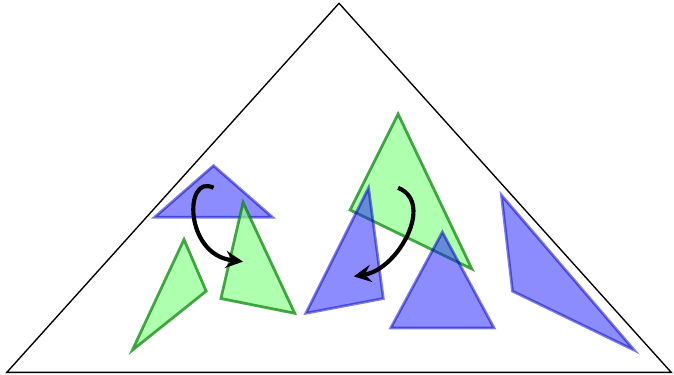}
\end{center}\vspace{-3ex}
\caption{\textit{Orthogonalization of parallel steps; the arrow indicates replacement.}}
\label{fig:orthogonalization:parallel}
\end{figure}
\end{proof}

\begin{definition}
Let $\pstep{\astep}{s}{t_1}$, $\pstep{\bstep}{s}{t_2}$ be parallel steps in a \woTRS{}.
The \emph{weakly orthogonal projection $\project{\astep}{\bstep}$ of $\astep$ over $\bstep$}
is defined as the orthogonal projection $\project{\astep'}{\bstep'}$
where $\pair{\astep'}{\bstep'}$ is the orthogonalization of $\astep$ and $\bstep$.
\end{definition}

\begin{remark}
The weakly orthogonal projection does not give rise to a residual system
in the sense of~\cite{terese:2003}.
The projection fulfils the three identities
  $\project{\astep}{\astep} \approx \unit$,
  $\project{\astep}{\unit} \approx \astep$, and
  $\project{\unit}{\astep} \approx \unit$,
but not the \emph{cube identity}
  $\project{(\project{\astep}{\bstep})}{(\project{\cstep}{\bstep})} \approx
    \project{(\project{\astep}{\cstep})}{(\project{\bstep}{\cstep})}$,
depicted in Figure~\ref{fig:cube}.

\begin{figure}[hpt!]
\begin{center}
\begin{tikzpicture}
  [inner sep=1mm,
   node distance=30mm,
   terminal/.style={
     rectangle,rounded corners=2.25mm,minimum size=4mm,
     very thick,draw=black!50,top color=white,bottom color=black!20,
   },
   >=stealth,
   red/.style={very thick,->},
   behind/.style={gray,thick},
   diag/.style={inner sep=.5mm}
  ]

  \node (s111) {};
  \node (s211) [right of=s111] {};
  \node (s121) [above of=s111] {};
  \node (s221) [above of=s211] {};
  \begin{scope}[node distance=22mm]
    \node (s112) [above right of=s111] {};
    \node (s212) [above right of=s211] {};
    \node (s122) [above right of=s121] {};
    \node (s222) [above right of=s221] {};
  \end{scope}

  \draw [red,behind] (s111) -- (s112); \node at ($(s111)!.5!(s112)$) [anchor=south east,diag] {$\astep$};
  \draw [red,behind,dotted] (s112) -- (s212); % \node at ($(s112)!.25!(s212)$) [anchor=south,behind] {$\project{\cstep}{\astep}$};
  \draw [red,behind,dotted] (s112) -- (s122); % \node at ($(s112)!.3!(s122)$) [anchor=east,behind] {$\project{\bstep}{\astep}$};

  \draw [red] (s111) -- (s211); \node at ($(s111)!.5!(s211)$) [anchor=north] {$\cstep$};
  \draw [red] (s111) -- (s121); \node at ($(s111)!.5!(s121)$) [anchor=east] {$\bstep$};

  \draw [red] (s211) -- (s212); \node at ($(s211)!.5!(s212)$) [anchor=north west,diag] {$\project{\astep}{\cstep}$};
  \draw [red] (s211) -- (s221); \node at ($(s211)!.4!(s221)$) [anchor=west] {$\project{\bstep}{\cstep}$};
  \draw [red] (s121) -- (s221); \node at ($(s121)!.35!(s221)$) [anchor=south] {$\project{\cstep}{\bstep}$};
  \draw [red] (s121) -- (s122); \node at ($(s121)!.5!(s122)$) [anchor=south east,diag] {$\project{\astep}{\bstep}$};

  \draw [red,dashed] (s221) -- (s222)
    node [midway,sloped,above] {$\project{(\project{\astep}{\bstep})}{(\project{\cstep}{\bstep})}$}
    node [midway,sloped,below] {$\project{(\project{\astep}{\cstep})}{(\project{\bstep}{\cstep})}$};

  \draw [red,dotted,thin,shorten >= 6mm] (s122) -- (s222);
  \draw [red,dotted,thin,shorten >= 6mm] (s212) -- (s222);
\end{tikzpicture}
\end{center}\vspace{-3ex}
\caption{\textit{Cube identity 
   $\project{(\project{\astep}{\bstep})}{(\project{\cstep}{\bstep})} \approx
    \project{(\project{\astep}{\cstep})}{(\project{\bstep}{\cstep})}$.}}
\label{fig:cube}
\end{figure}
\end{remark}

\begin{lemma}\label{lem:proj:lift}
  Let $\pstep{\astep}{s}{t_1}$, $\pstep{\bstep}{s}{t_2}$ be parallel steps in a \woTRS{} $\atrs$.
  Let $d_\astep$ and $d_\bstep$ be the minimal depth of a step in $\astep$ and $\bstep$, respectively.
  Then the minimal depth of the weakly orthogonal projections
  $\project{\astep}{\bstep}$ and $\project{\bstep}{\astep}$
  is greater or equal $\bfunap{\min}{d_\astep}{d_\bstep}$.
  If $\atrs$ contains no collapsing rules
  then the minimal depth of $\project{\astep}{\bstep}$ and $\project{\bstep}{\astep}$ is greater or equal $\bfunap{\min}{d_\astep}{d_\bstep+1}$
  and $\bfunap{\min}{d_\bstep}{d_\astep+1}$, respectively.
\end{lemma}

\begin{proof}
  Immediate from the definition of the orthogonalization (for overlaps the innermost redex is chosen)
  and the fact that in the orthogonal projection
  a non-collapsing rule applied at depth $d$
  can lift nested redexes at most to depth $d+1$ (but not above).
\end{proof}

\begin{lemma}[Strip/Lift Lemma]\label{lem:strip}
Let $\atrs$ be a \woTRS{}, 
$\aseq \funin s \to^\alpha t_1$ a rewrite sequence,
and $\pstep{\astep}{s}{t_2}$ a parallel rewrite step.
Let $d_\aseq$ and $d_\bseq$ be the minimal depth of a step in $\aseq$
and $\astep$, respectively.
Then there exist a term $u$, a rewrite sequence $\bseq \funin t_2 \to^{\le \omega} u$
and a parallel step $\pstep{\bstep}{t_1}{u}$
such that the minimal depth of the rewrite steps in $\bseq$ and $\bstep$ is $\bfunap{\min}{d_\aseq}{d_\bseq}$;
see Figure~\ref{fig:strip:collapse}.

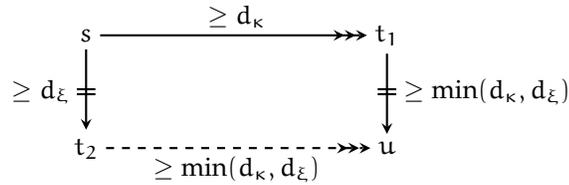
\begin{figure}[hpt!]
\begin{center}\hspace{1cm}
\begin{tikzpicture}
  [inner sep=1mm,
   node distance=15mm,
   terminal/.style={
     rectangle,rounded corners=2.25mm,minimum size=4mm,
     very thick,draw=black!50,top color=white,bottom color=black!20,
   },
   >=stealth,
   red/.style={thick,->>},
   infred/.style={
     thick,
     shorten >= 1mm,
     decoration={
       markings,
       mark=at position -3mm with {\arrow{stealth}},
       mark=at position -1.5mm with {\arrow{stealth}},
       mark=at position 1 with {\arrow{stealth}}
     },
     postaction={decorate}
   }]

  \node (s) {$s$};
  \node (t1) [node distance=40mm,right of=s] {$t_1$};
  \node (t2) [below of=s] {$t_2$};
  \node (u) [below of=t1] {$u$};

  \draw [infred] (s) -- (t1);
  \draw [red,->] (s) -- (t2);
  \draw [infred,dashed] (t2) -- (u);
  \draw [red,->] (t1) -- (u);

  \draw [thick] ($(s)!.5!(t2) + (-.75ex,-.20ex)$) -- ($(s)!.5!(t2) + (.75ex,-.20ex)$);
  \draw [thick] ($(s)!.5!(t2) + (-.75ex,.20ex)$) -- ($(s)!.5!(t2) + (.75ex,.20ex)$);
  \draw [thick] ($(t1)!.5!(u) + (-.75ex,-.20ex)$) -- ($(t1)!.5!(u) + (.75ex,-.20ex)$);
  \draw [thick] ($(t1)!.5!(u) + (-.75ex,.20ex)$) -- ($(t1)!.5!(u) + (.75ex,.20ex)$);

  \node at ($(s)!.5!(t1)$) [anchor=south] {$\ge d_\aseq$};
  \node at ($(s)!.5!(t2)$) [anchor=east,xshift=-.7ex] {$\ge d_\bseq$};
  \node at ($(t2)!.5!(u)$) [anchor=north] {$\ge \bfunap{\min}{d_\aseq}{d_\bseq}$};
  \node at ($(t1)!.5!(u)$) [anchor=west,xshift=.7ex] {$\ge \bfunap{\min}{d_\aseq}{d_\bseq}$};
\end{tikzpicture}
\end{center}\vspace{-3ex}
\caption{\textit{Strip/Lift Lemma with collapsing rules.}}
\label{fig:strip:collapse}
\end{figure}

If additionally $\atrs$ contains no collapsing rules, then
the minimal depth of a step in $\bseq$ and $\bstep$
is $\bfunap{\min}{d_\aseq}{d_\bseq+1}$ and $\bfunap{\min}{d_\bseq}{d_\aseq+1}$, respectively. See also Figure~\ref{fig:strip}.

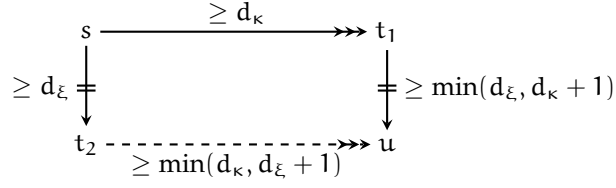
\begin{figure}[hpt!]
\begin{center}\hspace{2cm}
\begin{tikzpicture}
  [inner sep=1mm,
   node distance=15mm,
   terminal/.style={
     rectangle,rounded corners=2.25mm,minimum size=4mm,
     very thick,draw=black!50,top color=white,bottom color=black!20,
   },
   >=stealth,
   red/.style={thick,->>},
   infred/.style={
     thick,
     shorten >= 1mm,
     decoration={
       markings,
       mark=at position -3mm with {\arrow{stealth}},
       mark=at position -1.5mm with {\arrow{stealth}},
       mark=at position 1 with {\arrow{stealth}}
     },
     postaction={decorate}
   }]

  \node (s) {$s$};
  \node (t1) [node distance=40mm,right of=s] {$t_1$};
  \node (t2) [below of=s] {$t_2$};
  \node (u) [below of=t1] {$u$};

  \draw [infred] (s) -- (t1);
  \draw [red,->] (s) -- (t2);
  \draw [infred,dashed] (t2) -- (u);
  \draw [red,->] (t1) -- (u);

  \draw [thick] ($(s)!.5!(t2) + (-.75ex,-.20ex)$) -- ($(s)!.5!(t2) + (.75ex,-.20ex)$);
  \draw [thick] ($(s)!.5!(t2) + (-.75ex,.20ex)$) -- ($(s)!.5!(t2) + (.75ex,.20ex)$);
  \draw [thick] ($(t1)!.5!(u) + (-.75ex,-.20ex)$) -- ($(t1)!.5!(u) + (.75ex,-.20ex)$);
  \draw [thick] ($(t1)!.5!(u) + (-.75ex,.20ex)$) -- ($(t1)!.5!(u) + (.75ex,.20ex)$);

  \node at ($(s)!.5!(t1)$) [anchor=south] {$\ge d_\aseq$};
  \node at ($(s)!.5!(t2)$) [anchor=east,xshift=-.7ex] {$\ge d_\bseq$};
  \node at ($(t2)!.5!(u)$) [anchor=north] {$\ge \bfunap{\min}{d_\aseq}{d_\bseq+1}$};
  \node at ($(t1)!.5!(u)$) [anchor=west,xshift=.7ex] {$\ge \bfunap{\min}{d_\bseq}{d_\aseq+1}$};
\end{tikzpicture}
\end{center}\vspace{-3ex}
\caption{\textit{Strip/Lift Lemma without collapsing rules.}}
\label{fig:strip}
\end{figure}
\end{lemma}

\begin{proof}
  By compression we may assume $\alpha \le \omega$ in $\aseq \funin s \to^{\le \omega} t_1$ 
  (note that, the minimal depth $d$ is preserved by compression).
  Let $\aseq \funin s \equiv s_0 \to s_1 \to s_2 \to \ldots$,
  and define $\bstep_0 = \bstep$.
  Furthermore, let $\seqpref{\aseq}{n}$ denote the prefix of $\aseq$ of length $n$, that is, $s_0 \to \ldots \to s_n$
  and let $\seqsuf{\aseq}{n}$ denote the suffix $s_n \to s_{n+1} \to \ldots$ of $\aseq$.
  We employ the projection of parallel steps to
  close the elementary diagrams with top $s_n \to s_{n+1}$ and left $\pstep{\bstep_n}{s_n}{s_n'}$,
  that is, 
  we construct the projections $\bstep_{i+1} = \project{\bstep_i}{(s_i \to s_{i+1})}$ (right)
  and $\project{(s_i \to s_{i+1})}{\bstep_i}$ (bottom).
  %For $n \in \nat$, $n \le \alpha$ let $\aseq_n$ denote the finite prefix of $\aseq$ of length $n$, that is, $\aseq_n \funin s \to^n s_n$.
  Then by induction on $n$ using Lemma~\ref{lem:proj:lift} there exists
  for every $1 \le n \le \alpha$
  a term $s_n'$, and parallel steps $\pstep{\astep_n}{s_n}{s_n'}$ and $s_{n-1}' \pred s_n'$.
  See Figure~\ref{fig:strip:proof} for an overview.

  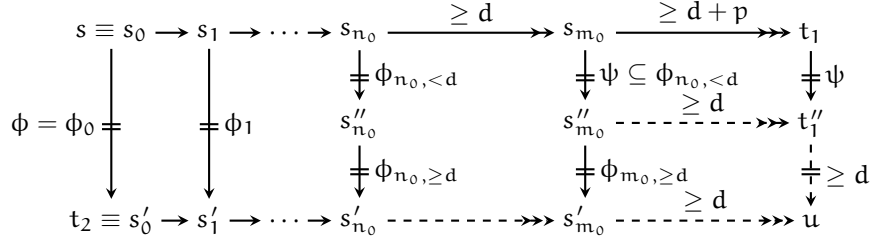
\begin{figure}[hpt!]
  \begin{center}
  \begin{tikzpicture}
    [inner sep=1mm,
    node distance=25mm,
    terminal/.style={
      rectangle,rounded corners=2.25mm,minimum size=4mm,
      very thick,draw=black!50,top color=white,bottom color=black!20,
    },
    >=stealth,
    red/.style={thick,->>},
    infred/.style={
      thick,
      shorten >= 1mm,
      decoration={
        markings,
        mark=at position -3mm with {\arrow{stealth}},
        mark=at position -1.5mm with {\arrow{stealth}},
        mark=at position 1 with {\arrow{stealth}}
      },
      postaction={decorate}
    }]
  
    \node (s) {$s \equiv s_0$};
    \node (s1) [node distance=13mm,right of=s] {$s_1$};
    \node (s2) [node distance=10mm,right of=s1] {$\ldots$};
    \node (sn0) [node distance=10mm,right of=s2] {$s_{n_0}$};
    \node (sm0) [node distance=30mm,right of=sn0] {$s_{m_0}$};
    \node (t1) [node distance=30mm,right of=sm0] {$t_1$};
    \node (t2) [below of=s] {$t_2 \equiv s_0'$};
    \node (s1') [below of=s1] {$s_1'$};
    \node (s2') [below of=s2] {$\ldots$};
    \node (sn0') [below of=sn0] {$s_{n_0}'$};
    \node (sm0') [below of=sm0] {$s_{m_0}'$};
    \node (u) [below of=t1] {$u$};

    \node (sn0'') [node distance=12mm,below of=sn0] {$s_{n_0}''$};
    \node (sm0'') [node distance=12mm,below of=sm0] {$s_{m_0}''$};
    \node (t1'') [node distance=12mm,below of=t1] {$t_1''$};

    \draw [red,->] (s) -- (s1);
    \draw [red,->] (s1) -- (s2);
    \draw [red,->] (s2) -- (sn0);
    \draw [red,->>] (sn0) -- (sm0) node [midway,above] {$\ge d$};
    \draw [infred] (sm0) -- (t1) node [midway,above] {$\ge d+p$};

    \draw [red,->] (s) -- (t2) node [midway,left,xshift=-.5ex] {$\astep = \astep_0$};
    \draw [thick] ($(s)!.5!(t2) + (-.75ex,-.20ex)$) -- ($(s)!.5!(t2) + (.75ex,-.20ex)$);
    \draw [thick] ($(s)!.5!(t2) + (-.75ex,.20ex)$) -- ($(s)!.5!(t2) + (.75ex,.20ex)$);

    \draw [red,->] (s1) -- (s1') node [midway,right,xshift=.5ex] {$\astep_1$};
    \draw [thick] ($(s1)!.5!(s1') + (-.75ex,-.20ex)$) -- ($(s1)!.5!(s1') + (.75ex,-.20ex)$);
    \draw [thick] ($(s1)!.5!(s1') + (-.75ex,.20ex)$) -- ($(s1)!.5!(s1') + (.75ex,.20ex)$);

    \draw [red,->] (sn0) -- (sn0'') node [midway,right,xshift=.5ex] {$\astep_{n_0,<d}$};
    \draw [thick] ($(sn0)!.5!(sn0'') + (-.75ex,-.20ex)$) -- ($(sn0)!.5!(sn0'') + (.75ex,-.20ex)$);
    \draw [thick] ($(sn0)!.5!(sn0'') + (-.75ex,.20ex)$) -- ($(sn0)!.5!(sn0'') + (.75ex,.20ex)$);
    \draw [red,->] (sn0'') -- (sn0') node [midway,right,xshift=.5ex] {$\astep_{n_0,\ge d}$};
    \draw [thick] ($(sn0'')!.5!(sn0') + (-.75ex,-.20ex)$) -- ($(sn0'')!.5!(sn0') + (.75ex,-.20ex)$);
    \draw [thick] ($(sn0'')!.5!(sn0') + (-.75ex,.20ex)$) -- ($(sn0'')!.5!(sn0') + (.75ex,.20ex)$);

    \draw [red,->] (sm0) -- (sm0'') node [midway,right,xshift=.5ex] {$\bstep \subseteq \astep_{n_0,<d}$};
    \draw [thick] ($(sm0)!.5!(sm0'') + (-.75ex,-.20ex)$) -- ($(sm0)!.5!(sm0'') + (.75ex,-.20ex)$);
    \draw [thick] ($(sm0)!.5!(sm0'') + (-.75ex,.20ex)$) -- ($(sm0)!.5!(sm0'') + (.75ex,.20ex)$);
    \draw [red,->] (sm0'') -- (sm0') node [midway,right,xshift=.5ex] {$\astep_{m_0,\ge d}$};
    \draw [thick] ($(sm0'')!.5!(sm0') + (-.75ex,-.20ex)$) -- ($(sm0'')!.5!(sm0') + (.75ex,-.20ex)$);
    \draw [thick] ($(sm0'')!.5!(sm0') + (-.75ex,.20ex)$) -- ($(sm0'')!.5!(sm0') + (.75ex,.20ex)$);

    \draw [red,->] (t1) -- (t1'') node [midway,right,xshift=.5ex] {$\bstep$};
    \draw [thick] ($(t1)!.5!(t1'') + (-.75ex,-.20ex)$) -- ($(t1)!.5!(t1'') + (.75ex,-.20ex)$);
    \draw [thick] ($(t1)!.5!(t1'') + (-.75ex,.20ex)$) -- ($(t1)!.5!(t1'') + (.75ex,.20ex)$);
    \draw [red,->,dashed] (t1'') -- (u) node [midway,right,xshift=.5ex] {$\ge d$};
    \draw [thick] ($(t1'')!.5!(u) + (-.75ex,-.20ex)$) -- ($(t1'')!.5!(u) + (.75ex,-.20ex)$);
    \draw [thick] ($(t1'')!.5!(u) + (-.75ex,.20ex)$) -- ($(t1'')!.5!(u) + (.75ex,.20ex)$);

    \draw [red,->] (t2) -- (s1');
    \draw [red,->] (s1') -- (s2');
    \draw [red,->] (s2') -- (sn0');
    \draw [infred,dashed] (sn0') -- (sm0');
    \draw [infred,dashed] (sm0'') -- (t1'')  node [midway,above] {$\ge d$};
    \draw [infred,dashed] (sm0') -- (u)  node [midway,above] {$\ge d$};
  \end{tikzpicture}
  \end{center}\vspace{-3ex}
  \caption{\textit{Strip/Lift Lemma, proof overview.}}
  \label{fig:strip:proof}
  \end{figure}

  We show that the rewrite sequence constructed at the bottom $s_0' \pred s_1' \pred \ldots$
  of Figures~\ref{fig:strip:collapse} and~\ref{fig:strip} is strongly convergent,
  and that the parallel steps $\astep_i$ have a limit for $i \to \infty$
  (parallel steps are always strongly convergent).

  Let $d \in \nat$ be arbitrary.
  By strong convergence of $\aseq$ there exists $n_0 \in \nat$ such that
  all steps in $\seqsuf{\aseq}{n_0}$ are at depth $\ge d$.
  Since $\astep_{n_0}$ is a parallel step there are only finitely many
  redexes $\astep_{n_0,<d} \subseteq \astep_{n_0}$ in $\astep_{n_0}$ rooted above depth $d$.
  By projection of $\astep_{n_0}$ along $\seqsuf{\aseq}{n_0}$
  no fresh redexes above depth $d$ can be created.
  The steps in $\astep_{n_0,<d}$ may be cancelled out due to overlaps,
  nevertheless, for all $m \ge n_0$ the set of steps above depth $d$ in $\astep_m$
  is a subset of $\astep_{n_0,<d}$.

  Let $p$ be the maximal depth of a left-hand side of a rule applied in $\astep_{n_0,<d}$.
  By strong convergence of $\aseq$ there exists $m_0 \ge n_0 \in \nat$ such that
  all steps in $\seqsuf{\aseq}{n_0}$ are at depth $\ge d+p$.
  As a consequence the steps $\bstep$ in $\astep_{m_0}$ rooted above depth $d$
  will stay fixed throughout the remainder of the projection.
  Then for all $m \ge m_0$ the parallel step $\astep_m$
  can be split into $\astep_m = s_m \pred_\bstep s_m'' \pred_{\astep_{m,\ge d}} s_m'$
  where $\astep_{m,\ge d}$ consists of the steps of $\astep_m$ at depth $\ge d$.
  Since $d$ was arbitrary, it follows that projection of $\astep$ over $\aseq$ has a limit.
  Moreover the steps of the projection of $\seqsuf{\aseq}{m_0}$ over $\astep_{m_0}$
  are at depth $\ge d + p - p = d$ since rules with pattern depth $\le p$ can lift steps by at most by $p$.
  Again, since $d$ was arbitrary, it follows that the projection of $\aseq$ over $\astep$ is strongly convergent.

  Finally, both constructed rewrite sequences (bottom and right) converge towards the same limit $u$
  since all terms $\{s_m', s_m'' \where m \ge m_0\}$ coincide up to depth $d-1$
  (the terms $\{s_m \where m \ge m_0\}$ coincide up to depth $d + p -1$ and the lifting effect of the steps $\astep_m$ is limited by $p$).
\end{proof}

\begin{theorem}\label{thm:cr}
Every \woTRS{} without collapsing rules is infinitary confluent.
\end{theorem}

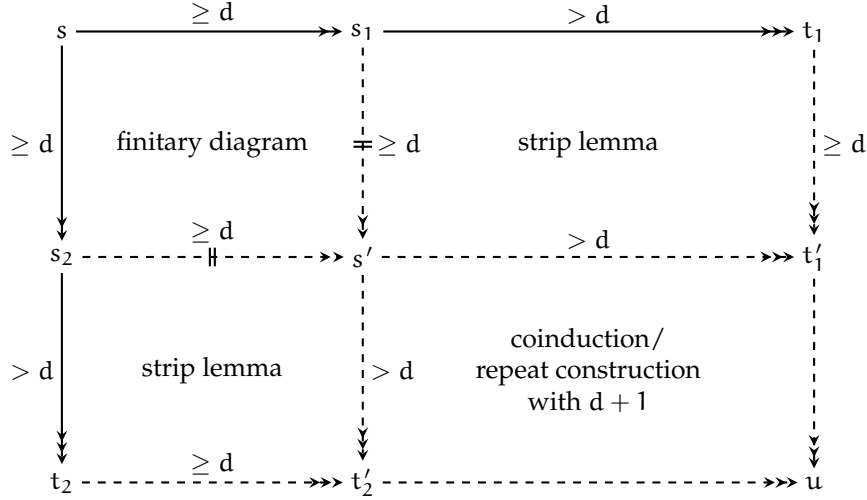
\begin{figure}[hpt!]
\begin{center}
\begin{tikzpicture}
  [inner sep=1mm,
   node distance=30mm,
   terminal/.style={
     rectangle,rounded corners=2.25mm,minimum size=4mm,
     very thick,draw=black!50,top color=white,bottom color=black!20,
   },
   >=stealth,
   red/.style={thick,->>},
   infred/.style={
     thick,
     shorten >= 1mm,
     decoration={
       markings,
       mark=at position -3mm with {\arrow{stealth}},
       mark=at position -1.5mm with {\arrow{stealth}},
       mark=at position 1 with {\arrow{stealth}}
     },
     postaction={decorate}
   }]

  \node (s) {$s$};
  \node (s1) [node distance=40mm,right of=s] {$s_1$};
  \node (t1) [node distance=60mm,right of=s1] {$t_1$};
  \node (s2) [below of=s] {$s_2$};
  \node (t2) [below of=s2] {$t_2$};
  \node (s') [below of=s1] {$s'$};
  \node (t1') [below of=t1] {$t_1'$};
  \node (t2') [below of=s'] {$t_2'$};
  \node (u) [below of=t1'] {$u$};

  \draw [red] (s) -- (s1);
  \draw [infred] (s1) -- (t1);
  \draw [red] (s) -- (s2);
  \draw [infred] (s2) -- (t2);

  \draw [red,dashed] (s1) -- (s');
  \draw [thick] ($(s1)!.5!(s') + (-.75ex,-.20ex)$) -- ($(s1)!.5!(s') + (.75ex,-.20ex)$);
  \draw [thick] ($(s1)!.5!(s') + (-.75ex,.20ex)$) -- ($(s1)!.5!(s') + (.75ex,.20ex)$);

  \draw [red,dashed] (s2) -- (s');
  \draw [thick] ($(s2)!.5!(s') + (-.20ex,-.75ex)$) -- ($(s2)!.5!(s') + (-.20ex,.75ex)$);
  \draw [thick] ($(s2)!.5!(s') + (.20ex,-.75ex)$) -- ($(s2)!.5!(s') + (.20ex,.75ex)$);

  \draw [infred,dashed] (s') -- (t1');
  \draw [infred,dashed] (t1) -- (t1');
  \draw [infred,dashed] (s') -- (t2');
  \draw [infred,dashed] (t2) -- (t2');
  \draw [infred,dashed] (t1') -- (u);
  \draw [infred,dashed] (t2') -- (u);

  \node at ($(s)!.5!(s1)$) [anchor=south] {$\ge d$};
  \node at ($(s)!.5!(s2)$) [anchor=east] {$\ge d$};
  \node at ($(s1)!.5!(t1)$) [anchor=south] {$>d$};
  \node at ($(s2)!.5!(t2)$) [anchor=east] {$>d$};
  \node at ($(s2)!.5!(s')$) [anchor=south,yshift=.5ex] {$\ge d$};
  \node at ($(s1)!.5!(s')$) [anchor=west,xshift=.5ex] {$\ge d$};

  \node at ($(s')!.5!(t1')$) [anchor=south] {$> d$};
  \node at ($(s')!.5!(t2')$) [anchor=west] {$> d$};

  \node at ($(t2)!.5!(t2')$) [anchor=south] {$\ge d$};
  \node at ($(t1)!.5!(t1')$) [anchor=west] {$\ge d$};

  \node at ($(s)!.5!(s')$) {finitary diagram};
  \node at ($(s')!.5!(t1)$) {strip lemma};
  \node at ($(s')!.5!(t2)$) {strip lemma};
  \node at ($(s')!.5!(u)$) {\parbox{3cm}{\centerline{coinduction/}\centerline{repeat construction}\centerline{with $d+1$}}};
\end{tikzpicture}
\end{center}\vspace{-3ex}
\caption{\textit{Infinitary confluence.}}
\label{fig:confluence}
\end{figure}

\begin{proof}
  An overview of the proof is given in Figure~\ref{fig:confluence}.
  Let $\aseq \funin s \to^\alpha t_1$ and $\bseq \funin s \to^\beta t_2$ be two rewrite sequences.
  By compression we may assume $\alpha \le \omega$ and $\beta \le \omega$.
  Let $d$ be the minimal depth of any rewrite step in $\aseq$ and $\bstep$.
  Then $\aseq$ and $\bseq$ are of the form
  $\aseq \funin s \to^* s_1 \to^{\le\omega} t_1$
  and
  $\bseq \funin s \to^* s_2 \to^{\le\omega} t_2$
  such that all steps in $s_1 \to^{\le\omega} t_1$ and $s_2 \to^{\le\omega} t_2$ at depth $> d$.
  
  Then $s \to^* s_1$ and $s \to^* s_2$ can be joined by finitary diagram completion employing the diamond property for parallel steps.
  If follows that there exists a term $s'$ and finite sequences of (possibly infinite) parallel steps $s_1 \pred^* s'$ and $s_2 \pred^* s'$
  all steps of which are at depth $\ge d$ (Lemma~\ref{lem:proj:lift}).
  We project 
  $s_1 \to^{\le\omega} t_1$ over $s_1 \pred^* s'$ 
  $s_2 \to^{\le\omega} t_2$ over $s_2 \pred^* s'$ 
  by repeated application of the Lemma~\ref{lem:strip},
  obtaining rewrite sequences 
  $t_1 \ired t_1'$,
  $s' \ired t_1'$,
  $t_2 \ired t_2'$, and
  $s' \ired t_2'$ with depth $\ge d$, $> d$, $\ge d$, and $> d$, respectively.
  As a consequence we have $t_1'$, $s'$ and $t_2'$ coincide up to (including) depth $d$.
  Recursively applying the construction to the rewrite sequences $s' \ired t_1'$ and $s' \ired t_2'$
  yields strongly convergent rewrite sequences 
  $t_2 \ired t_2' \ired t_2'' \ired \ldots$ and $t_1 \ired t_1' \ired t_1'' \ired \ldots$
  where the terms $t_1^{(n)}$ and $t_2^{(n)}$ coincide up to depth $d + n -1$.
  Thus these rewrite sequences converge towards the same limit $u$.
\end{proof}

We consider an example to illustrate that non-collapsingness is a necessary condition for Theorem~\ref{thm:cr}.
\begin{example}\label{ex:collapse}
  Let $\atrs$ be a TRS over the signature $\{f,a,b\}$ consisting of the rule:
  \begin{align*}
    \bfunap{f}{x}{y} &\to x
  \end{align*}
  Then, using a self-explaining recursive notation, 
  the term $s = \bfunap{f}{\bfunap{f}{s}{b}}{a}$ rewrites in $\omega$ many steps to 
  $t_1 = \bfunap{f}{t_1}{a}$ as well as $t_2 = \bfunap{f}{t_2}{b}$ which have no common reduct.
  The TRS $\atrs$ is weakly orthogonal (even orthogonal) but not confluent.
\end{example}

\section{The Diamond Property for Developments}\label{sec:diamond}

We prove that infinitary developments in weakly orthogonal TRSs 
without collapsing rules have the diamond property.
For this purpose we establish an orthogonalization algorithm for co-initial developments,
that is, we make the developments orthogonal to each other by elimination of overlaps.
Since overlapping steps in weakly orthogonal TRSs have the same targets, 
we can replace  one by the other.
The challenge is to reorganize the steps in such a way that no new overlaps are created.

\begin{figure}[hpt!]
\begin{center}
  \includegraphics{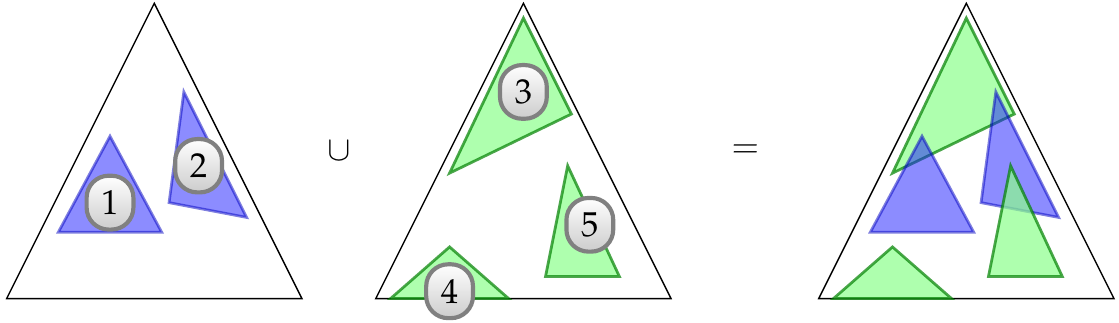}
\end{center}\vspace{-3ex}
\caption{\textit{Orthogonalization in a weakly orthogonal TRS.}}
\label{fig:orthogonalization}
\end{figure}

Consider for example Figure~\ref{fig:orthogonalization}, where the redexes 2 and 3 overlap with each other.
When trying to solve this overlap,
we have to be careful since replacing the redex 2 by 3 as well as 3 by 2 creates new conflicts.

The case of finitary weakly orthogonal rewriting is treated by \mbox{van Oostrom} 
and \mbox{de Vrijer} in~\cite[Theorem 8.8.23]{terese:2003}. 
They employ an inside-out algorithm, that is, inductively extend an orthogonalization 
of the subtrees to the whole tree. The basic observation is that you overcome the difficulties 
pointed out above by starting at the bottom of the tree and solving overlaps by choosing 
the deeper (innermost) redex. 

\begin{example}
  We consider Figure~\ref{fig:orthogonalization} and 
  apply the orthogonalization algorithm from \cite[Theorem 8.8.23]{terese:2003}. 
  We start at the bottom of the tree. 
  The first overlap we find is between the redexes 2 and 5; 
  this is removed by replacing 2 with 5. 
  Then the overlap between 2 and 3 has also disappeared.
  The only remaining overlap is between the redexes 3 and 1. Hence we replace 3 by 1.
  As result we obtain two orthogonal developments $\{1,5\}$ and $\{1,4,5\}$.
\end{example}

Note that the above algorithm does not carry over to the case of infinitary developments
since we may have infinite chains of overlapping redexes and thus have no bottom to start at.
This is illustrated in Figure~\ref{fig:chain}.
\begin{figure}[hpt!]
  \begin{center}
    \includegraphics{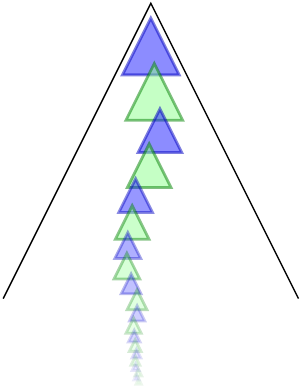}
  \end{center}
  \caption{\textit{Infinite chain of overlaps.}}
  \label{fig:chain}
\end{figure}

\begin{example}{
  \renewcommand{\a}{\funap{A}}
  As an example where such an infinite chain of overlaps arises
  we consider the TRS $R$ consisting of the rule:
  \begin{align*}
     \a{\a{\a{x}}} \to \a{x}
  \end{align*}
  together with two developments of blue and green redexes in the term $A^{\!\omega}$:
  \begin{center}
    \includegraphics{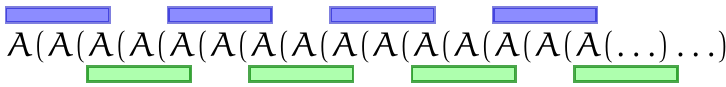}
  \end{center}
  The blue redexes are marked by overlining, the green redexes by underlining.
}\end{example}

\begin{definition}\label{def:overlap}\
  Let $R$ be a \woTRS{}, and $t \in \iter{\asig}$ a term.
  For redexes $u$ and $v$ in $t$ we write $u \poverlap v$ 
  if the pattern of $u$ and $v$ overlap in $t$.
\end{definition}

\begin{definition}
  Let $R$ be a \woTRS{}, $t \in \iter{\asig}$ a term, and $\sigma, \delta$ two developments 
  of sets of redexes $U$ and $V$ in $t$, respectively.
  We call $\sigma$ and $\delta$ \emph{orthogonal (to each other)}
  if for all redexes $u \in U$ and $v \in V$ with $u \poverlap v$
  we have that $u = v$ (redexes are the same, that is, with respect to the same rule and position).

  An \emph{orthogonalization $\pair{\sigma'}{\delta'}$ of $\sigma$ and $\delta$}
  consists of orthogonal developments $\sigma'$, $\delta'$ 
  of redexes in $t$ such that the results (targets) of $\sigma'$ and $\delta'$
  coincide with the results of $\sigma$ and $\delta$, respectively.
\end{definition}

The overlap relation $\poverlap$ is symmetric, and hence
the redexes form clusters with respect to the transitive closure $\poverlap^*$.
If such a cluster contains parallel redexes, 
then the redex-cluster is called \emph{Y-cluster} in \cite{kete:klop:oost:2004a}.
The redexes in a \mbox{Y-cluster} we call Y-redexes; they can be defined as follows.

\begin{definition}\label{def:noeffect}
  Let $R$ be a \woTRS{}, and $t \in \iter{\asig}$ a term.
  A redex $u$ in $t$ is called \emph{Y-redex}
  if there exist redexes $v_1$, $v_2$ at disjoint positions in $t$ 
  such that $u \poverlap^* v_1$ and $u \poverlap^* v_2$ 
  (see also Figure~\ref{fig:cases:orthogonalization}, cases (ii) and (iv)).
\end{definition}

At first sight one might expect that Y-redexes
are due to trivial rules of the form $\ell \to r$ with $\ell \equiv r$.
However, the following example illustrates that this is in general not the case
(for another example see \cite[p.508, middle]{terese:2003}).

\begin{example}\label{ex:noeffect}
  Let $\atrs$ consist of the following rules:
  \begin{align*}
    (\rho_1)\;\; \funap{f}{\bfunap{g}{x}{y}} &\to \funap{f}{\bfunap{g}{y}{x}} &
    (\rho_2)\;\;\bfunap{g}{a}{a} &\to \bfunap{g}{a}{a} &
    (\rho_3)\;\;a &\to a
  \end{align*}
  We consider the term $t \equiv \funap{f}{\bfunap{g}{a}{a}}$
  which contains a $\rho_1$-redex $u_{\posemp}$ at the root,
  a $\rho_2$-redex $u_{1}$ at position $1$,
  and two $\rho_3$-redexes $u_{11}$ and $u_{12}$ at position $11$ and $12$, respectively.
  We have $u_{\posemp} \poverlap u_{1}$, $u_{1} \poverlap u_{11}$ and $u_{1} \poverlap u_{12}$.
  Since $u_{11}$, and $u_{12}$ are at disjoint positions, 
  it follows that $u_{\posemp}$ is a Y-redex.
  However, the rule $\rho_1$ permutes its subterms, and 
  thus in general may very well have an effect.
\end{example}

Despite the above example, 
it is always safe to drop Y-redexes from developments without changing the outcome of the development.
This result is implicit in \cite{kete:klop:oost:2004a}.
In particular in \cite[Remark~4.38]{kete:klop:oost:2004a}
it is mentioned that Y-clusters are a generalisation of Takahashi-configurations.

\begin{lemma}\label{lem:noeffect}
Let $R$ be a \woTRS{}, $t \in \iter{\asig}$ a term, 
and $U$ a set of non-overlapping redexes in $t$ which have a development.
If $u \in U$ is a Y-redex, then the development
of $U$ results in the same term as the development of $U \setminus \{u\}$.
\end{lemma}

\begin{proof}
  Since $u$ is a Y-redex there exist redexes 
  $v_1,v_2$ at disjoint positions in $t$ such that
  $v_1 \poverlap^* u \poverlap^* v_2$.
  By weak orthogonality overlapping redexes have the same effect.
  Whenever $w_1 \poverlap w_2 \poverlap w_3$,
  it follows that $w_1$ has the same effect as $w_3$.
  Consequently all redexes in the Y-cluster $Y_u = \{v \where u \poverlap^* v\}$ of $u$
  have the same effect.
  However, $v_1$ and $v_2$ are at disjoint positions, and
  thus it follows that ($*$) contraction of any redex in $Y_u$ leaves $t$ unchanged.

  We now develop $U$ using an innermost strategy (from bottom to top).
  Then reducing redexes below the pattern of $Y_u$ leaves the Y-cluster
  unchanged, and if a redex $v$ of $Y_u$ (among which is $u$) is reduced,
  then by ($*$) the term is left unchanged, and by innermost strategy
  there are no redexes nested in $v$ which could be influenced (copied/ deleted).
  Thus contracting $u$ has no effect.
\end{proof}

We also give the following alternative proof using results from \cite{kete:klop:oost:2004a}.

\begin{proof}[Proof (by Vincent van Oostrom).]
We cut out the Y-cluster (union of the redex patterns), and introduce distinct new variables
for the cut-off subterms. Contracting a redex in the Y-cluster leaves the Y-cluster unchanged.
In particular no subterms matched by variables are moved, copied or deleted.
As a consequence contracting $u$ can only affect redexes which are in the Y-cluster of $u$,
and thus in turn their contraction does not change the term.
It follows that $u$ can be dropped from the development.
\end{proof}

We now devise a top--down orthogonalization algorithm.
Roughly, we start at the top of the term and replace overlapping redexes with the outermost one.
However, care has to be taken in situations as depicted in Figure~\ref{fig:orthogonalization}.

\begin{theorem}\label{thm:orthogonalization}
  Let $R$ be a \woTRS{},
  $t \in \iter{\asig}$ a (possibly infinite) term, and $\sigma, \delta$ two developments 
  of sets of redexes $U$ and $V$, respectively. Then there exists an orthogonalization of $U$ and $V$.
\end{theorem}

\begin{proof}
  We obtain an orthogonalization of $U$ and $V$ as the limit of the following process.
  If there are no overlaps between $U$ and $V$, then we are finished.
  Here, by overlap we mean non-identical redexes whose patterns overlap.
  Otherwise, if there exist overlaps, let $u \in (U \cup V)$ be a topmost redex 
  (that is, having minimal depth)
  among the redexes which have an overlap.
  Without loss of generality (by symmetry) we assume that $u \in U$ and let $v \in V$
  be a topmost redex among the redexes in $V$ overlapping $u$.
  We distinguish the following cases:
  \begin{figure}[hpt!]
	\begin{center}
          \includegraphics{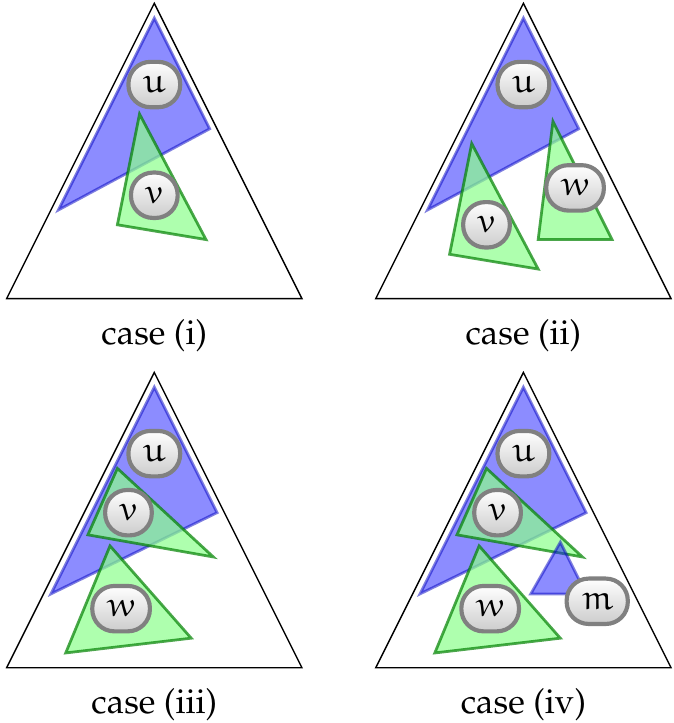}
	\end{center}
	\caption{\textit{Case distinction for the orthogonalization algorithm.}}
	\label{fig:cases:orthogonalization}
  \end{figure}
  \begin{enumerate}[(i)]
	 \item If $v$ is the only redex in $V$ that overlaps with $u$, case (i) of Figure~\ref{fig:cases:orthogonalization},
		   then we can safely replace $v$ by $u$.
  \end{enumerate}
  \noindent
  Otherwise there is a redex $w \in V$, $w \ne v$ and $w$ overlaps $u$.
  \begin{enumerate}[(i)]
	\setcounter{enumi}{1}
	\item Assume that $v$ and $v'$ are at disjoint positions, case (ii) of Figure~\ref{fig:cases:orthogonalization}.
		  Then $u$, $v$, $w$ are Y-redexes and can be dropped from $U$ and $V$ by Lemma~\ref{lem:noeffect}.
  \end{enumerate}
  \noindent
  Otherwise, $v$ and $v'$ are not disjoint, and then $w$ must be nested inside $v$.
  \begin{enumerate}[(i)]
	\setcounter{enumi}{2}
	\item If $u$ is the only redex from $U$ overlapping $v$, case (iii) of Figure~\ref{fig:cases:orthogonalization}, then we can replace $u$ by $v$.
	\item In the remaining case there must be a redex $m \in U$, $m \ne u$ and $m$ overlaps with the redex $v$, 
		  see case (iv) of Figure~\ref{fig:cases:orthogonalization}.
		  Since $U$ and $V$ are developments $u$ cannot overlap with $m$, and $v$ cannot overlap with $w$.
		  We have that $w$ is nested in $v$, both overlapping $u$, but $m$ is below the pattern of $u$, overlapping $v$.
		  Hence $w$ and $m$ must be at disjoint positions 
		  ($v$ cannot tunnel through $w$ to touch $m$); this has also been shown in \cite{kete:klop:oost:2004a}.
		  Then by Lemma~\ref{lem:noeffect} all redexes $u$, $m$, $v$ and $w$ are Y-redexes and can be removed.
  \end{enumerate}
  
  We have shown that it is always possible to solve outermost conflicts without creating fresh ones. Hence we can push the conflicts down to infinity and thereby obtain the orthogonalization of $U$ and $V$.
\end{proof}

We obtain the diamond property as corollary.
\begin{corollary}\label{cor:diamond}
  For every \woTRS{} without collapsing rules, (infinite) developments have the diamond property.
\end{corollary}

\begin{proof}
  Let $\sigma, \delta$ be two coinitial developments $t_1 \stackrel{\sigma}{\redi} s \stackrel{\delta}{\red} t_2$.
  Then by Theorem~\ref{thm:orthogonalization} there exists an orthogonalization $\pair{\sigma'}{\delta'}$
  of $\sigma$, $\delta$.
  The orthogonal projections $\project{\sigma'}{\delta'}$ and $\project{\delta'}{\sigma'}$
  are developments again, which are strongly convergent since the rules are not collapsing.
  Hence $t_1 \stackrel{\project{\delta'}{\sigma'}}{\red} s' \stackrel{\project{\sigma'}{\delta'}}{\redi} t_2$.
\end{proof}

Note that in Corollary~\ref{cor:diamond} the non-collapsingness is a necessary condition.
To see this, reconsider Example~\ref{ex:collapse} and observe that the non-confluent derivations are developments.

In a similar vein, we can prove the triangle property for infinitary weakly orthogonal
developments without collapsing rules, but we will postpone this to future work.

\section{Conclusions}

We have shown the failure of $\UNinf$ for weakly orthogonal TRSs
in the presence of two collapsing rules.
For \woTRS{s} without collapsing rules we establish that $\CRinf$ (and hence $\UNinf$) holds.
This result is optimal in the sense that
allowing only one collapsing rule may invalidate $\CRinf$.

However, Question~\ref{q:one} remains open:
\begin{quote}
  \emph{Does $\UNinf$ hold for \woTRS{s} with one collapsing rule?}
\end{quote}
Furthermore, we have shown that infinitary developments in \woTRS{s} without collapsing rules
have the diamond property. In general this property fails already in the presence
of only one collapsing rule.

\paragraph{Acknowledgement.}
  We want to thank Vincent van~Oostrom for many helpful remarks
  and pointers to work on weakly orthogonal TRSs.

\bibliography{main}

\end{document}